\pdfoutput=1
\RequirePackage{ifpdf}
\ifpdf 
\documentclass[pdftex]{sigma}
\else
\documentclass{sigma}
\fi

\usepackage{accents}
\def \h#1{\widehat{#1}}
\def \t#1{\widetilde{#1}}
\def \b#1{\overline{#1}}
\def \ub#1{\underline{#1}}
\def \c#1{\accentset{\circ}{#1}}
\def \th#1{\widehat{\widetilde{#1}}}

\def \tb#1{{\widetilde{\overline{#1}}}}

\makeatletter
\def\wideutilde{\underaccent{{\cc@style\undertilde{\mskip6mu}}}}
\def\ut{\underaccent{{\cc@style\undertilde{\mskip6mu}}}}
\makeatother

\numberwithin{equation}{section}

\newtheorem{Theorem}{Theorem}[section]
\newtheorem{Lemma}[Theorem]{Lemma}
 { \theoremstyle{definition}
\newtheorem{Remark}[Theorem]{Remark} }

\begin{document}
\allowdisplaybreaks

\newcommand{\arXivNumber}{1702.01266}

\renewcommand{\thefootnote}{}

\renewcommand{\PaperNumber}{078}

\FirstPageHeading

\ShortArticleName{Rational Solutions to the ABS List: Transformation Approach}

\ArticleName{Rational Solutions to the ABS List:\\ Transformation Approach\footnote{This paper is a~contribution to the Special Issue on Symmetries and Integrability of Dif\/ference Equations. The full collection is available at \href{http://www.emis.de/journals/SIGMA/SIDE12.html}{http://www.emis.de/journals/SIGMA/SIDE12.html}}}

\Author{Danda ZHANG and Da-Jun ZHANG}
\AuthorNameForHeading{D.D.~Zhang and D.-J.~Zhang}
\Address{Department of Mathematics, Shanghai University, Shanghai 200444, P.R.~China}
\Email{\href{mailto:zhangdd@shu.edu.cn}{zhangdd@shu.edu.cn}, \href{mailto:djzhang@staff.shu.edu.cn}{djzhang@staff.shu.edu.cn}}

\ArticleDates{Received March 21, 2017, in f\/inal form September 26, 2017; Published online October 02, 2017}

\Abstract{In the paper we derive rational solutions for the lattice potential modif\/ied Korteweg--de Vries equation, and Q2, Q1($\delta$), H3($\delta$), H2 and H1 in the Adler--Bobenko--Suris list. B\"acklund transformations between these lattice equations are used. All these rational solutions are related to a unif\/ied $\tau$ function in Casoratian form which obeys a bilinear superposition formula.}


\Classification{35Q51; 35Q55}

\renewcommand{\thefootnote}{\arabic{footnote}}
\setcounter{footnote}{0}

\section{Introduction}\label{sec-1}

In recent decades the research of discrete integrable systems has undergone rapid progress (see~\cite{HJN-book-2016} and the references therein). As a new concept, multidimensional consistency, allowing suitable lattice equations to be embedded into a higher-dimensional space in a consistent way, has played an important role in the research of quadrilateral equations \cite{ABS-CMP-2003,ABS-FAA-2009,BS-IMRN-2002,Hie-JPA-2011, Nij-PLA-2002,NijW-GMJ-2001}. Quadrilateral equations that are consistent around the cube (CAC) with additional restriction ($D_4$ symmetry and tetrahedron property) were searched and classif\/ied by Adler, Bobenko and Suris (ABS)~\cite{ABS-CMP-2003} and in their list only 9 equations are included: Q4, Q3($\delta$), Q2, Q1($\delta$), A2, A1($\delta$), H3($\delta$), H2 and H1. All these equations have been solved from dif\/ferent approaches \cite{AHN-JPA-2007,AHN-JPA-2008,AN-CMP-2010,HZ-JPA-2009, NAH-JPA-2009,SZ-SIGMA-2011,ZH-AIP-2010}.

As for rational solutions, which are solutions expressed by fractions of polynomials, in general, such type of solutions can be derived from soliton solutions through a special limit procedure (or a Taylor expansion), which corresponds to a way to generate multiple zero eigenvalues for certain spectral problems (see \cite{AblS-1978,Zha-2006} as examples). For the $\delta$-dependent equations in the ABS list, for example, H3($\delta$) and Q1($\delta$), the existence of $\delta$ (i.e., $\delta \neq 0$) plays a crucial role \cite{SZ-SIGMA-2011} in the procedure of obtaining rational solutions from their soliton solutions. For H1 which is independent of $\delta$, its rational solutions were obtained recently by making use of the Hirota--Miwa equation and a continuous auxiliary variable \cite{FZS-ZNA-2016}. Besides, as a generic (2+1)-D bilinear model, polynomial solutions of the Hirota--Miwa equation have been derived from several ways and presented via dif\/ferent forms \cite{GRPSW-JPA-2007,MKNO-PLA-1997}.

In this paper we systematically construct rational solutions for the ABS list by means of B\"acklund transformations (BTs). A fundamental role playing in the paper is the lattice potential modif\/ied Korteweg--de Vries (lpmKdV) equation. There is a non-auto BT which connects the lpmKdV equation and Q1(0) (also known as the lattice Schwarzian Korteweg--de Vries equation and cross-ratio equation). The two equations and their BT constitute a consistent triplet, say, viewing the BT as a two-component system, then the compatibility of each component yields a~lattice equation of another component which is in the triplet. This means any pair of solutions of the BT provide solutions to the two equations that the BT connects. Details will be shown in Sections~\ref{sec-3-1} and \ref{sec-3-2} on how such a consistent triplet works in generating rational solutions. We also make use of non-auto BTs between equations in the ABS list~\cite{Atk-JPA-2008}. Starting from the lpmKdV equation and Q1(0), rational solutions of Q2, Q1($\delta$), A1($\delta$), H3($\delta$), H2 and H1 in the ABS list can be derived through the map:
\begin{figure}[h!]
$$\begin{array}{@{}ccccccc}
{\rm H2} & & {\rm Q1(0)} & \longrightarrow & \mathrm{Q1(\delta)} & \longrightarrow & {\rm Q2}\\
\big \updownarrow && \big \updownarrow && \big \downarrow &&\\
{\rm H1} & \longleftarrow & {\rm lpmKdV} & & \mathrm{A1(\delta)} & \longleftrightarrow & \mathrm{H3(\delta)}
\end{array}
$$
\caption{A map for generating rational solutions.}\label{fig1}
\end{figure}

\noindent
In the map the double-head arrow means the two equations it connects and their BT form a~consistent triplet.

Moreover, we f\/ind all the obtained rational solutions are related to a unif\/ied $\tau$ function in Casoratian form which obeys a bilinear superposition formula (see~\eqref{f-itera}). Compared with those rational solutions of H3($\delta$) and Q1($\delta$) derived in \cite{SZ-SIGMA-2011}, here we obtain new solutions. In fact, we will see that rational solutions of Q1($\delta$) can explicitly be expressed through the rational solutions of Q1(0). Similar results hold for~H3($\delta$) as well.

The paper is organized as follows. In Section~\ref{sec-2} as preliminary we list quadrilateral equations that we consider in the paper and some notations. Then in Sections~\ref{sec-3} and~\ref{sec-4} we derive some rational solutions for the equations listed in Section~\ref{sec-2}. In Section~\ref{sec-5} rational solutions in Casoratian form are proved. Finally in Section~\ref{sec-6} we give conclusions.

\section{Preliminary}\label{sec-2}

We list quadrilateral equations that we consider in the paper:
\begin{alignat}{3}
& {\rm H1}\colon \quad && (\t{u}-\h{u})\big(\th{u}-u\big)=q-p,&\label{H1}\\
& {\rm H2}\colon \quad && (\t{v}-\h{v})\big(v-\h{\t{v}}\big)+(q-p)\big(v+\t{v}+\h{v}+\h{\t{v}}\big)+q^2-p^2=0,& \label{H2}\\
& {\rm lpmKdV}\colon \quad && a\big(V\t{V}-\h{V}\h{\t{V}}\big)-b\big(V\h{V}-\t{V}\h{\t{V}}\big)=0,& \label{lpmkdv}\\
& {\rm H3}\colon \quad && a\big(Z\t{Z}+\h{Z}\h{\t{Z}}\big)-b\big(Z\h{Z}+\t{Z}\h{\t{Z}}\big)+2\delta\big(a^2-b^2\big)=0,& \label{H3}\\
& {\rm Q1(0)}\colon \quad && p(v-\h{v})\big(\t{v}-\h{\t{v}}\big)-q(v-\t{v})\big(\h{v}-\h{\t{v}}\big)=0,&\label{Q1-0}\\
& \mathrm{Q1(\delta)}\colon \quad && p(u-\h{u})\big(\t{u}-\h{\t{u}}\big)-q(u-\t{u})\big(\h{u}-\h{\t{u}}\big)+\delta^2 pq(p-q)=0,& \label{Q1-d}\\
& \mathrm{A1(\delta)}\colon \quad && p(z+\h{z})\big(\t{z}+\h{\t{z}}\big)-q(z+\t{z})\big(\h{z}+\h{\t{z}}\big)-\delta^2pq(p-q)=0,& \label{A1}\\
& {\rm Q2}\colon \quad && p(w-\h{w})\big(\t{w}-\th{w}\big)-q(w-\t{w})\big(\h{w}-\th{w}\big)+pq(p-q)\big(w+\t{w}+\h{w}+\th{w}\big)&\nonumber\\
&&& \qquad {}-pq(p-q)\big(p^2-pq+q^2\big)=0.& \label{Q2}
\end{alignat}
Here we use conventional notations $\t{u}\doteq u_{n+1,m}$, $\h{u}\doteq u_{n,m+1}$. In the above equations, $p$ and $a$ are spacing parameters of $n$-direction and $q$ and $b$ are of $m$-direction; $\delta$ is an arbitrary constant.

Casoratian is a discrete version of Wronskian. Suppose that a basic column vector is
\begin{gather*}
\psi(n,m,l)=\big(\psi_1(n,m,l),\psi_2(n,m,l),\dots,\psi_{N}(n,m,l)\big)^{\rm T}.
\end{gather*}
Introduce a shift operator $E_{\nu}$ to denote
\begin{gather*}
E^{s}_n \psi\equiv \psi(n+s,m,l),\qquad E^{s}_m\psi\equiv \psi(n,m+s,l),\qquad E^{s}_{l}\psi\equiv \psi(n,m,l+s).
\end{gather*}
Then a $N$th-order Casoratian w.r.t.\ $E_{l}$-shift is def\/ined by
\begin{gather*}
\big|\psi, E_l\psi,E_l^2\psi,E_l^3\psi,\dots, E_l^{N-1}\psi\big|,
\end{gather*}
and usually is compactly written as (cf.~\cite{FN-PLA-1983-1})
\begin{gather*}
\big|\h{N-1}|=|0,1,2,\dots,N-1\big|.
\end{gather*}
Another notation which is often used is $|\h{N-2},N|=|0,1,\dots,N-2,N|$.

Besides, in Wronskian/Casoratian verif\/ication of solutions to a bilinear equation, the equation is usually reduced to a~Laplace expansion of a zero-valued $2N\times 2N$ determinant. The expansion is described as
\begin{Lemma}[\cite{FN-PLA-1983-1}]\label{L:lap}
Suppose that $\mathbf{B}$ is an $N\times(N-2)$ matrix and $\mathbf{a}$, $\mathbf{b}$, $\mathbf{c}$, $\mathbf{d}$ are $N$th-order
column vectors, then
\begin{gather*}
|\mathbf{B},\mathbf{a}, \mathbf{b}||\mathbf{B}, \mathbf{c},\mathbf{d}|
-|\mathbf{B}, \mathbf{a}, \mathbf{c}||\mathbf{B},\mathbf{b},\mathbf{d}|
+|\mathbf{B},\mathbf{a},\mathbf{d}||\mathbf{B},\mathbf{b},\mathbf{c}|=0.
\end{gather*}
\end{Lemma}

\section{Rational solutions to lpmKdV, Q1, H3 and Q2}\label{sec-3}

In this section we f\/irst investigate relation between the lpmKdV equation and Q1(0). Such a~relation will be used to construct rational solutions to not only the two equations themselves but also to Q1($\delta$), H3($\delta$) and Q2.

\subsection{Solution sequence of Q1(0) and lpmKdV}\label{sec-3-1}

Q1(0) is the equation \eqref{Q1-d} with $\delta=0$. 
Between \eqref{Q1-0} 
and the lpmKdV equation \eqref{lpmkdv} there is a non-auto B\"acklund transformation \cite{NAH-JPA-2009}
\begin{gather}
 \t{v}-v=a V\t{V},\qquad \h{v}-v=b V\h{V}, \label{BT-S-mkdv}
\end{gather}
where
\begin{gather}
p=a^2,\qquad q=b^2.\label{pq-ab-Q1}
\end{gather}
Equations \eqref{Q1-0}, \eqref{lpmkdv} and \eqref{BT-S-mkdv} constitute a consistent triplet in the following sense: as an equation set, the compatibility of $\t{\h v}=\h{\t v}$ and $\t{\h V}=\h{\t V}$ respectively yield~\eqref{lpmkdv} and~\eqref{Q1-0}.

Such an consistency can be used to construct solutions for equation \eqref{Q1-0} and \eqref{lpmkdv}:
\begin{Lemma}\label{lem-1}
With the consistent triplet composed of \eqref{Q1-0}, \eqref{lpmkdv} and \eqref{BT-S-mkdv}, we have the fol\-lowing:
\begin{enumerate}\itemsep=0pt
\item[$(1)$] starting from any solution $v$ of \eqref{Q1-0}, by integration through \eqref{BT-S-mkdv}, the resulted $V$ solves equation \eqref{lpmkdv}, and vice versa;
\item[$(2)$] any solution pair $(v, V)$ of \eqref{BT-S-mkdv} gives a solution $v$ to \eqref{Q1-0} and $V$ to \eqref{lpmkdv}.
\end{enumerate}
\end{Lemma}

Further than that, we have
\begin{Lemma}\label{lem-2} For an arbitrary solution pair $(v, V)$ of \eqref{BT-S-mkdv} where $V\neq 0$, function $V_1=v/V$ solves the {\rm lpmKdV} equation \eqref{lpmkdv}.
\end{Lemma}
\begin{proof} Substituting $V_1={v}/{V}$ into \eqref{lpmkdv} and making use of relation~\eqref{BT-S-mkdv}, it is easy to check~$V_{1}$ satisf\/ies~\eqref{lpmkdv}.
\end{proof}

This lemma provides an approach to generate a sequence of solution pairs of the BT \eqref{BT-S-mkdv}.
\begin{Theorem}\label{thm-1} For any solution pair $(v_N,V_N)$ of the BT \eqref{BT-S-mkdv}, define
\begin{subequations}\label{iterat}
\begin{gather}
V_{N+1}=\frac{v_N}{V_N},\label{V-N+1}
\end{gather}
and it solves the {\rm lpmKdV} equation \eqref{lpmkdv}. Next, the BT system
\begin{gather}
\t{v}_{N+1}-v_{N+1}=a V_{N+1}\t{V}_{N+1},\qquad \h{v}_{N+1}-v_{N+1}=b V_{N+1}\h{V}_{N+1},\label{BT-S-mkdv-N}
\end{gather}
\end{subequations}
$($i.e., \eqref{BT-S-mkdv}$)$ determines a function $v_{N+1}$ that satisfies equation {\rm Q1(0)} \eqref{Q1-0}. $v_N$ and $v_{N+1}$ obey the relation
\begin{gather}
(\t{v}_{N+1}-v_{N+1})(\t v_N-v_N)=a^2 v_{N}\t{v}_{N},\qquad (\h{v}_{N+1}-v_{N+1})(\h v_N-v_N)=b^2 v_{N}\h{v}_{N},\label{BT-Q10}
\end{gather}
which is an auto BT of {\rm Q1(0)}.
\end{Theorem}

\begin{proof}The f\/irst part of the theorem holds due to Lemma~\ref{lem-2}. For the second part, since $(v_N, V_N)$ is a solution pair of the BT \eqref{BT-S-mkdv}, they have compatibility $\th A= \t{\h A}$, and so does $V_{N+1}$. Then, on the basis of Lemma~\ref{lem-1}, $v_{N+1}$ def\/ined by~\eqref{BT-S-mkdv-N} solves~Q1(0). For the relation~\eqref{BT-Q10}, substituting~\eqref{V-N+1} into~\eqref{BT-S-mkdv-N} and making use of~\eqref{BT-S-mkdv} with $(v,V)=(v_N, V_N)$, we arrive at~\eqref{BT-Q10}, which provides an auto BT for~Q1(0). In fact, there is a non-auto BT~\cite{Atk-JPA-2008}
\begin{gather}
(u-\t u)(v-\t v)=p\big(v\t v-\delta^2\big),\qquad (u-\h u)(v-\h v)=q\big(v\h v-\delta^2\big),\label{BT:Q1Q1}
\end{gather}
to map $v$ to $u$ from Q1(0) to Q1($\delta$). It holds as well for the degenerated case $\delta=0$, in which both $v$ and $u$ are solutions of Q1(0).
\end{proof}

\subsection{Rational solutions of Q1(0) and lpmKdV}\label{sec-3-2}

Theorem \ref{thm-1} describes an iterative mechanism to generate new solutions for Q1(0) and the lpmKdV equation. Thus, if we start from a simple solution pair, e.g., $(v_1=an+bm+\gamma_1, V_1=1)$, we can generate a sequence of rational solutions to Q1(0) and the lpmKdV equation. Some low order solutions in this sequence are
\begin{subequations}
\begin{gather}
 v_1=x_1,\qquad V_1=1,\label{v1}\\
 v_2=\frac{1}{3}\big(x_1^3-x_3\big),\qquad V_2=x_1,\label{v2}\\
 v_3=\frac{1}{x_1}\left(\frac{1}{45}x_1^6-\frac{1}{9}x_1^3x_3+\frac{1}{5}x_1x_5-\frac{1}{9}x_3^2\right), \qquad V_3=\frac{x_1^3-x_3}{3x_1},\label{v3}\\
 v_4=\frac{3}{x_1^3-x_3}\left(\frac{1}{4725}x_1^{10}-\frac{1}{315}x_1^7x_3+\frac{1}{75}x_1^5x_5-\frac{1}{27}x_1x_3^3\right.\nonumber\\
\left. \hphantom{v_4=}{} -\frac{1}{25}x_5^2 +\frac{1}{15}x_1^2x_3x_5-\frac{1}{21}x_1^3x_7+\frac{1}{21}x_3x_7\right),\nonumber\\
 V_4=\frac{3}{x_1^3-x_3}\left(\frac{1}{45}x_1^6-\frac{1}{9}x_1^3x_3+\frac{1}{5}x_1x_5-\frac{1}{9}x_3^2\right),
\end{gather}\label{v123}
\end{subequations}
where
\begin{gather}
x_i=a^in+b^im+\gamma_i,\qquad \gamma_i\in \mathbb{C}, \quad i=1,2,\dots.
\label{xi}
\end{gather}
We note that $\{V_N\}$ are dif\/ferent from the rational solutions of the lpmKdV equation obtained in~\cite{MKNO-PLA-1997} as a~reduction of the Hirota--Miwa equation.

In the following we prove that if we start from~\eqref{v123}, all the solutions generated from~\eqref{iterat} are meaningful. First, let us look at non-zero property.

\begin{Lemma}\label{lem-3} Suppose $a>0$, $b>0$, $v_N(0,0)>0$, and we restrict $(n,m)$ in the first quadrant $\{n\geq 0,\, m\geq 0\}$. Then~$v_N$ and~$V_N$ generated from~\eqref{iterat} with~\eqref{v1} satisfy $v_N > 0$, $V_N > 0$.
\end{Lemma}
\begin{proof}Obviously, under assumption of the lemma, from \eqref{v1} we have $v_1>0$, $V_1>0$ and $V_2=v_1/V_1>0$. Then, suppose that $v_N>0$, $V_N>0$ and consequently $V_{N+1}=v_N/V_N>0$. Next, from~\eqref{BT-S-mkdv-N} we have
\begin{gather*}v_{N+1}(n+1,m)- v_{N+1}(n,m)>0,\qquad v_{N+1}(n,m+1)- v_{N+1}(n+1,m)>0.\end{gather*}
This implies
\begin{gather*}v_{N+1}(n,m)>v_{N+1}(n-1,m)> v_{N+1}(n-1,m-1)>\cdots > v_{N+1}(0,0).\end{gather*}
If we take ``integration'' constant $v_{N+1}(0,0)>0$, then $v_{N+1}(n,m)$ must be positive in quadrant $\{n\geq 0,\, m\geq 0\}$.
\end{proof}

Next, we observe that in $v_1$, $v_2$ and $v_3$ the order of leading terms (in terms of~$x_1$) are respectively~1,~3 and~5. Now we prove all the $v_N$ def\/ined through~\eqref{iterat} with~\eqref{v1} are distinct in the sense of having dif\/ferent leading orders in terms of~$x_1$.

\begin{Lemma}\label{lem-4} $v_N$ has a leading order $2N-1$ in terms of $x_1$ and $V_N$ has a leading order $N-1$ in the same sense.
\end{Lemma}

\begin{proof} From \eqref{v123} we can suppose the lemma is correct up to some integer $N$. Then one can f\/ind
\begin{gather*}V_{N+1}=\frac{v_N}{V_N} \sim O\big(x_1^{N}\big),\end{gather*}
and from \eqref{BT-S-mkdv-N}
\begin{gather*}\t v_{N+1}- v_{N+1}=a V_N\t V_{N} \sim O\big(x_1^{2N}\big),\qquad \h v_{N+1}- v_{N+1}=b V_N\h V_{N} \sim O\big(x_1^{2N}\big),\end{gather*}
which means $v_{N+1}\sim O\big(x_1^{2N+1}\big)$. Based on mathematical induction, the lemma holds.
\end{proof}

We conclude the following.
\begin{Theorem}\label{thm-2} The iteration relation \eqref{iterat} is meaningful in terms of generating distinct rational solutions from initial solutions~\eqref{v1} for {\rm Q1(0)} and the {\rm lpmKdV} equation. These solutions are positive at least on the first quadrant $\{n\geq 0,\, m\geq 0\}$ if we take $a>0$, $b>0$, $v_N(0,0)>0$.
\end{Theorem}

Note that not all solutions can be ef\/fectively iterated through~\eqref{iterat}. For example,
\begin{gather*}
 v_1=\alpha^n\beta^m,\qquad V_1=v_1^{\frac{1}{2}}
\end{gather*}
are conf\/ined in \eqref{iterat} due to $V_2=v_1/V_1=V_1$. Here the parameterizations for~$a$ and~$b$ are $a^2=p={(1-\alpha)^2}/{\alpha}$, $b^2=q={(1-\beta)^2}/{\beta}$.

\subsection[Solutions to Q1($\delta$)]{Solutions to Q1($\boldsymbol{\delta}$)}\label{sec-3-3}

The iteration \eqref{iterat} can be extended to $N\leq 0$.

\begin{Lemma}\label{lem-5}
Define
\begin{gather}
v_{-N}=-\frac{1}{v_{N+1}}, \qquad V_{-N}=(-1)^{N+1}\frac{1}{V_{N+2}}, \qquad \text{for} \quad N \geq 0.\label{vN-}
\end{gather}
Then the iteration relation \eqref{iterat} can be extended to $N\in \mathbb{Z}$.
\end{Lemma}

This lemma can be checked directly.

Note that the extension does not lead to new solutions to Q1(0) and the lpmKdV equation because these two equations are invariant under transformations of type $u \to c/u$. However, the extension does bring more rational solutions to Q1($\delta$).

\begin{Theorem}\label{thm-3}
For the pair $(v_N, V_N)$ determined by \eqref{iterat} with $N\in \mathbb{Z}$, function
\begin{gather}
u_N=v_{N}+\frac{\delta^2}{v_{N-2}}, \qquad N\in \mathbb{Z}, \label{u-N}
\end{gather}
gives a sequence of solutions to {\rm Q1($\delta$)}. $u_N$ and $v_{N-1}$ are connected via
\begin{gather}\label{BT-Q1Q10-N}
\t{u}_N-u_N=\frac{a^2(v_{N-1}\t{v}_{N-1}-\delta^2)}{\t{v}_{N-1}-v_{N-1}},\qquad
\h{u}_N-u_N=\frac{b^2(v_{N-1}\h{v}_{N-1}-\delta^2)}{\h{v}_{N-1}-v_{N-1}},
\end{gather}
which is the non-auto BT \eqref{BT:Q1Q1} between {\rm Q1($\delta$)} and {\rm Q1(0)}. Also, \eqref{u-N} agrees with the chain
\begin{gather}
 u_1~ \overset{~~\delta=0~~}{-\!-\!\!\!\longrightarrow} ~ v_1
~ \overset{{\rm BT}~\eqref{BT-Q1Q10-N}}{-\!-\!\!\!\longrightarrow} ~u_2
 \overset{~~\delta=0~~}{-\!-\!\!\!\longrightarrow} ~v_2
~ \overset{{\rm BT}~\eqref{BT-Q1Q10-N}}{-\!-\!\!\!\longrightarrow} ~u_3 ~ \cdots\cdots.\label{chain}
\end{gather}
\end{Theorem}

\begin{proof}We only need to prove $u_N$ def\/ined by \eqref{u-N} satisf\/ies \eqref{BT-Q1Q10-N}. Since $\{v_N\}$ obey the BT~\eqref{BT-Q10}, using which we can f\/ind
\begin{gather*}\t v_k-v_k=\frac{a^2 \t v_{k-1} v_{k-1}}{\t v_{k-1}-v_{k-1}},\qquad
\frac{1}{\t v_{k-2}}-\frac{1}{v_{k-2}}= \frac{-a^2}{\t v_{k-1}-v_{k-1}}.\end{gather*}
Then, from \eqref{u-N} by direct calculation we immediately have
\begin{gather*}\t{u}_N-u_N=\t v_N-v_N+\delta^2\left(\frac{1}{\t v_{N-2}}-\frac{1}{v_{N-2}}\right) =\frac{a^2(v_{N-1}\t{v}_{N-1}-\delta^2)}{\t{v}_{N-1}-v_{N-1}},\end{gather*}
which coincides with the f\/irst equation in \eqref{BT-Q1Q10-N}. Similarly we can f\/ind $\h{u}_N-u_N$ satisf\/ies the second equation in~\eqref{BT-Q1Q10-N} as well.
\end{proof}

Formula \eqref{u-N} provides an explicit relation between solutions of Q1($\delta$) and Q1(0), where $\{v_N\}$ is a sequence generated from~\eqref{iterat}. For $v_N$ given as in \eqref{v123}, some rational solutions of Q1($\delta$) generated from~\eqref{u-N} are
\begin{gather}
u_1=x_1-\frac{1}{3}\delta^2\big(x_1^3-x_3\big),\nonumber\\ 
u_2=\frac{1}{3}\big(x_1^3-x_3\big)-\delta^2 x_1,\label{u2}\\
u_3=\frac{1}{x_1}\left(\frac{1}{45}x_1^6-\frac{1}{9}x_1^3x_3+\frac{1}{5}x_1x_5-\frac{1}{9}x_3^2+\delta^2\right),\nonumber
\end{gather}
where we have made use of relation \eqref{vN-} in order to get $v_{-1}$ and $v_0$.

Here we give two remarks.

\begin{Remark}\textit{There are some overlaps $($dual forms$)$ in the chain \eqref{u-N}.} Note that Q1($\delta$) equation is formally invariant under transformation f\/irst replacing $u$ with $\varepsilon \delta^2 u$ and then $\delta$ with $1/{\delta}$ where $\varepsilon =\pm 1$, by which \eqref{u-N} is transformed into its dual form
\begin{gather*}
u_N=\varepsilon \left(\delta^2 v_{N}+\frac{1}{v_{N-2}}\right),
\end{gather*}
which gives a sequence of solutions to Q1($\delta$) as well. By the relation \eqref{vN-} given in Lemma~\ref{lem-5}, $u_N$ and $u_{3-N}$ in \eqref{u-N} are dual forms of each other.
\end{Remark}

\begin{Remark}\textit{Not all the rational solutions of {\rm Q1($\delta$)} are included in the chain \eqref{u-N}.}
We give two exceptions. One is
\begin{subequations}
\begin{gather}
u=\alpha n+\beta m+\gamma,\label{u-Q1-excp-1}
\end{gather}
where $\gamma$ is a constant and $\alpha$, $\beta$ are def\/ined by parametrization
\begin{gather}
 p=\frac{c_0}{a^2-\delta^2}, \qquad q=\frac{c_0}{b^2-\delta^2},\qquad \alpha=pa,\qquad \beta=qb,
 \label{u-Q1-excp-1-para}
\end{gather}
with arbitrary constant $c_0$,
\end{subequations}
and the other is
\begin{gather}
 u= \delta x^2_1+\delta\gamma_0,\label{u-Q1-excp-2}
\end{gather}
where $x_1$ is def\/ined as~\eqref{xi}, $p$, $q$ are parameterized as in \eqref{pq-ab-Q1} and $\gamma_0$ is a constant.
\end{Remark}

\subsection{Rational solutions to Q2}\label{sec-3-4}

We make use of a non-auto BT between Q1($\delta$) \eqref{Q1-d} and Q2 \eqref{Q2} to derive rational solutions of~Q2. The BT reads~\cite{Atk-JPA-2008}
\begin{gather}
\delta(u-\t{u})(w-\t{w})=p(2u\t{u}-\delta^2w-\delta^2\t{w})+\delta p^2(u+\t{u}+\delta p),\nonumber\\
\delta(u-\h{u})(w-\h{w})=q(2u\h{u}-\delta^2w-\delta^2\h{w})+\delta q^2(u+\h{u}+\delta q).
\label{BT-Q1Q2}
\end{gather}

When $u$ is given by \eqref{u-Q1-excp-1} with parametrization \eqref{u-Q1-excp-1-para}, from \eqref{BT-Q1Q2} we can f\/ind
\begin{gather*}
w=\frac{u^2}{\delta^2}-\frac{c_0 u}{\delta^3}+\frac{c_0^2 }{2\delta^4}
+\left(\frac{a-\delta}{a+\delta}\right)^n\left(\frac{b-\delta}{b+\delta}\right)^m\gamma_0,\qquad u=\alpha n+\beta m+\gamma,
\end{gather*}
where $\gamma_0$ is a constant. This is not a pure rational solution.

For $u$ def\/ined in \eqref{u-Q1-excp-2} with parametrization \eqref{pq-ab-Q1}, from \eqref{BT-Q1Q2} we f\/ind
\begin{gather*}
 w= \frac{1}{5}x_1^4+\frac{2}{3}\gamma_0x_1^2+\frac{4\gamma_0x_3}{3x_1}+\frac{4x_5}{5x_1}+\gamma_0^2. 
\end{gather*}

For $u=u_2$ given by \eqref{u2} with parametrization \eqref{pq-ab-Q1}, from \eqref{BT-Q1Q2} we f\/ind
\begin{gather*}
 w= \frac{1}{45\delta(x_1-\delta)}\big[30\delta^3\big(x_1^3-x_3\big) -15\delta^2\big(x_1^4+2x_1x_3\big) -3\delta\big(x_1^5-10x_1^2x_3-6x_5\big)\\
\hphantom{w=}{} +2x_1^6+18x_1x_5-10x_3^2-10x_1^3x_3\big].\label{w3-Q2}
\end{gather*}

\subsection[Solutions to H3($\delta$)]{Solutions to H3($\boldsymbol{\delta}$)}\label{sec-3-5}

To obtain solutions to H3($\delta$) we make use of A1($\delta$)~\eqref{A1}. Solutions to A1($\delta$) can be obtained from those of Q1($\delta$)~\eqref{Q1-d} through transformation~\cite{ABS-CMP-2003}
\begin{gather*}
z=(-1)^{n+m}u.
\end{gather*}
Similar to the triplet composed by \eqref{Q1-0}, \eqref{lpmkdv} and \eqref{BT-S-mkdv}, A1($\delta$) \eqref{A1} with parametrization \eqref{pq-ab-Q1}, H3($\delta$)~\eqref{H3} and their non-auto BT~\cite{Atk-JPA-2008}
\begin{gather}\label{BT-A1-H3}
 {\t{z}+z}-\delta a^2= aZ\t{Z},\qquad {\h{z}+z}-\delta b^2= bZ\h{Z}
\end{gather}
constitute a consistent triplet, i.e., compatibility $\th{z}=\t{\h z}$ in \eqref{BT-A1-H3} requires $Z$ satisf\/ies \eqref{H3} and $\th{Z}=\t{\h Z}$ requires $z$ satisf\/ies \eqref{A1}. Such a consistency leads to
\begin{Lemma}\label{lem-6}\quad
\begin{enumerate}\itemsep=0pt
\item[$(1)$] Starting from any solution $z$ of \eqref{A1}, by integration through \eqref{BT-A1-H3}, the resulted $Z$ sol\-ves~\eqref{H3}, and vice versa.
\item[$(2)$] any solution pair $(z, Z)$ of \eqref{BT-A1-H3} gives a solution $z$ to \eqref{A1} and $Z$ to \eqref{H3}.
\end{enumerate}
\end{Lemma}

Solution sequences of A1($\delta$) and H3($\delta$) are then given as follows.

\begin{Theorem}\label{thm-4}
For $v_N$ and $V_N$ constructed in Theorem~{\rm \ref{thm-1}} and Lemma~{\rm \ref{lem-5}}, function
\begin{gather}\label{z-N}
 z_N=(-1)^{n+m}\left(v_{N}+\frac{\delta^2}{v_{N-2}}\right), \qquad N\in \mathbb{Z},
\end{gather}
solves {\rm A1($\delta$)} \eqref{A1}, and
\begin{gather}
Z_N=(-1)^{\frac{n+m}{2}+\frac{1}{4}}\left(V_{N}+\frac{(-1)^{n+m}\delta}{V_{N-1}}\right), \qquad N\in \mathbb{Z},\label{Z-N}
\end{gather}
solves {\rm H3($\delta$)} \eqref{H3}.
\end{Theorem}

\begin{proof}Since $u_N$ def\/ined in \eqref{u-N} solves Q1($\delta$), it is obvious that \eqref{z-N} provides a solution to~A1($\delta$). Besides, making use of iterative relation \eqref{iterat}, one can f\/ind \eqref{z-N} and \eqref{Z-N} provide a~solution pair to \eqref{BT-A1-H3}, which proves the present theorem.
\end{proof}

Here we list some solutions for H3($\delta$),
\begin{gather*}
Z_1=(-1)^{\frac{n+m}{2}+\frac{1}{4}}\big(1-(-1)^{n+m}\delta x_1\big),\\
Z_2=(-1)^{\frac{n+m}{2}+\frac{1}{4}}\big(x_1+(-1)^{n+m}\delta\big),\\
Z_3=(-1)^{\frac{n+m}{2}+\frac{1}{4}}\frac{x_1^3-x_3+3(-1)^{n+m}\delta}{3x_1}.
\end{gather*}

Similar to Q1($\delta$), there are also overlaps (dual forms) in the chain \eqref{Z-N} for H3($\delta$). H3($\delta$)~is formally invariant under transformation f\/irst $Z \to \varepsilon \delta^{-1}(-1)^{n+m} Z$ and then $\delta \to -\delta^{-1}$ with $\varepsilon =\pm 1$, by which \eqref{Z-N} becomes
\begin{gather*}
Z_N=\varepsilon(-1)^{\frac{n+m}{2}+\frac{1}{4}}\left(\delta(-1)^{n+m} V_{N}-\frac{1}{V_{N-1}}\right).
\end{gather*}
Thus, $Z_N$ and $Z_{3-N}$ in \eqref{Z-N} are dual forms of each other in light of relation~\eqref{vN-}.

\section{Solutions of H1 and H2}\label{sec-4}

In this section we f\/irst derive solutions of H1 using a relation between H1 and the lpmKdV equation. Then from H1 we derive solutions of H2.

\subsection{H1}\label{sec-4-1}

There is a non-auto BT \cite{HJN-book-2016}
\begin{gather}\label{BT-H1-mkdv}
 \t{u}-\h{u}=\frac{b\t{V}-a\h{V}}{abV},\qquad \h{\t{u}}-u=\frac{b\h{\t{V}}+aV}{ab\h{V}}
\end{gather}
to connect H1($u$) \eqref{H1} and the lpmKdV($V$) equation~\eqref{lpmkdv}. We use it to derive solutions for H1 on the basis of the following fact.

\begin{Lemma}\label{lem-7} When $V$ solves the {\rm lpmKdV} equation~\eqref{lpmkdv}, function $u$ defined by~\eqref{BT-H1-mkdv} satis\-fies~{\rm H1}~\eqref{H1} with parametrization
\begin{gather}
p=-1/a^2,\qquad q=-1/b^2. \label{pq-ab-H1}
\end{gather}
\end{Lemma}

\begin{proof}We rewrite the lpmKdV equation \eqref{lpmkdv} as
\begin{gather}
 \big(b\t{V}-a\h{V}\big)\big(b\th{V}+aV\big)=\big(b^2-a^2\big)V\h{V}. \label{lpmkdv-2}
\end{gather}
Then, multiplying both equations in~\eqref{BT-H1-mkdv} and making use of~\eqref{lpmkdv-2} we immediately reach~H1~\eqref{H1} provided~$p$,~$q$ are parameterized as~\eqref{pq-ab-H1}.
\end{proof}

From \eqref{BT-H1-mkdv} we f\/ind
\begin{gather}\label{u-V}
 \t{\t{u}}-u=\frac{V+\t{\t{V}}}{a\t{V}},\qquad \h{\h{u}}-u=\frac{V+\h{\h{V}}}{b\h{V}},
\end{gather}
of which we make use to derive $u$ from known $V$.

For $V_1=1$ and $V_2=x_1$ given in \eqref{v123}, we derive same solution for H1,
\begin{subequations}\label{u123-H1}
\begin{gather}
 u_1=u_2=x_{-1},
\end{gather}
where $x_{-1}$ follows the def\/inition \eqref{xi} with $i=-1$. For $V_3$ and $V_4$ given in \eqref{v123}, we respectively f\/ind
\begin{gather}
 u_3=x_{-1}-\frac{1}{x_1}, \label{u3}
\end{gather}
and
\begin{gather}
 u_4=x_{-1}-\frac{3 x_1^2}{x_1^3-x_3}. \label{u4}
\end{gather}
\end{subequations}

Note that the lpmKdV equation \eqref{lpmkdv} is invariant under $V\to \frac{1}{V}$. So we can replace~$V$ by~$1/V$ in \eqref{u-V} and get
\begin{gather*}
 \t{\t{u}}-u=\frac{\t{V}}{a}\left(\frac{1}{V}+\frac{1}{\t{\t{V}}}\right),\qquad \h{\h{u}}-u=\frac{\h{V}}{b}\left(\frac{1}{V}+\frac{1}{\h{\h{V}}}\right).
\end{gather*}
One may wonder if the above relation can be used to generate more solutions for H1. However, making use of iterative relations~\eqref{iterat} we f\/ind
\begin{gather*}\t{\t{u}}_{N+1}-u_{N+1} =\frac{V_{N+1}+\t{\t{V}}_{N+1}}{a\t{V}_{N+1}}=\frac{\t{V}_N}{a}\left(\frac{1}{V_N}+\frac{1}{\t{\t{V}}_N}\right),
\end{gather*}
and a same formula for $(\,\h{~~},\, b)$, which means $V\to \frac{1}{V}$ does not lead to new solutions for H1. This can also explain the fact $u_1=u_2$ due to $V_1=\frac{1}{V_1}=1$.

\subsection{H2}\label{sec-4-2}

Again, H2 \eqref{H2}, H1 \eqref{H1} and their non-auto BT \cite{Atk-JPA-2008}
\begin{gather}\label{BT-H1-H2}
 v+\t{v}+p=2u\t{u},\qquad v+\h{v}+q=2u\h{u}
\end{gather}
constitute a consistent triplet with parametrization \eqref{pq-ab-H1}. Then, from solutions \eqref{u123-H1} of H1 and BT \eqref{BT-H1-H2}, we f\/ind the following rational solutions for H2:
\begin{gather*}
 v_1=v_2=x_{-1}^2,\qquad
 v_3=x_{-1}^2-\frac{2x_{-1}}{x_1},\qquad
 v_4=x_{-1}^2 -\frac{6 x_1 (x_1 x_{-1}-1)}{x_1^3-x_3}.
\end{gather*}

\section{Rational solutions in determinant form}\label{sec-5}

From the previous section it is understood that the sequence $\{V_N\}$ plays a crucial role in constructing solutions in the whole paper. With regard to rational solutions, it is hard to do ``integration'' from \eqref{BT-S-mkdv-N} to get high order~$v_N$ and consequently it is dif\/f\/icult to get high order~$V_N$. In this section we aim to construct Casoratian expressions for~$v_N$ and~$V_N$, as well as rational solutions of other equations.

\subsection[Bilinear relation of $V_N$ and $v_N$]{Bilinear relation of $\boldsymbol{V_N}$ and $\boldsymbol{v_N}$}\label{sec-5-1}

We express
\begin{subequations}\label{v-P}
\begin{gather}
V_N=\frac{P_{N-1}}{P_{N-2}}
\end{gather}
and it then follows from \eqref{V-N+1} that
\begin{gather}
v_N=\frac{P_{N}}{P_{N-2}}.
\end{gather}
\end{subequations}
From $V_2=x_1$ we introduce
\begin{gather*}
P_0=1,\qquad P_1=x_1,
\end{gather*}
and from \eqref{v123} we f\/ind successively
\begin{gather*}
P_2=\frac{x_1^3-x_3}{3},\\
P_3=\frac{1}{45}x_1^6-\frac{1}{9}x_1^3x_3+\frac{1}{5}x_1x_5-\frac{1}{9}x_3^2,\\
P_4=\frac{1}{4725}x_1^{10}-\frac{1}{315}x_1^7x_3+\frac{1}{75}x_1^5x_5-\frac{1}{27}x_1x_3^3-\frac{1}{25}x_5^2
+\frac{1}{15}x_1^2x_3x_5-\frac{1}{21}x_1^3x_7+\frac{1}{21}x_3x_7,
\end{gather*}
where $P_4$ is obtained from the relation $\frac{P_4}{P_3}=V_5=\frac{v_4}{V_4}$.

Viewing \eqref{v-P} as transformations, the BT \eqref{BT-S-mkdv-N} yields
\begin{gather}\label{P-iterat}
 \t{\b{P}} \ub P-\b{P}\t{\ub{P}}=aP\t{P},\qquad \h{\b{P}}\ub{P}-\b{P}\h{\ub{P}}=b P\h{P},
\end{gather}
where
\begin{gather*}P\doteq P_N,\qquad \b P \doteq P_{N+1},\qquad \ub P\doteq P_{N-1}.\end{gather*}
This is a bilinear system for polynomials $\{P_N\}$. Note that based on~\eqref{vN-} the relations~\eqref{v-P} and~\eqref{P-iterat} can be extended to $N\in \mathbb{Z}$ by def\/ining
\begin{gather}
P_{-N}=(-1)^{[\frac{N}{2}]}P_{N-1},\label{P-N-minus}
\end{gather}
where $[\,\cdot\,]$ denotes the greatest integer function.

In Section~\ref{sec-5-3} we will give a Casoratian form of $P$. To achieve that, we make use of H1.

\subsection{Casoratian form of rational solutions of H1}\label{sec-5-2}

For H1 \eqref{H1}, using 3D consistency we have its BT
\begin{gather}\label{BT-H1}
(\t{u}-\b{u})\big(\t{\b{u}}-u\big)= a^{-2}-k^{-2},\qquad (\h{u}-\b{u})\big(\h{\b{u}}-u\big)=b^{-2}-k^{-2},
\end{gather}
where $\b u$ stands for a new solution of H1, we adopt parametrization~\eqref{pq-ab-H1} and the arbitrary number $k^{-2}=r$ acts as a ``soliton number'' which leads to a new soliton (cf.~\cite{HZ-JPA-2009}). Now we remove the term $k^{-2}$ from~\eqref{BT-H1}, i.e., taking $r=0$, and consequently we have
\begin{gather}\label{BT-H1-R}
(\t{u}-\b{u})\big(\t{\b{u}}-u\big)=a^{-2},\qquad (\h{u}-\b{u})\big(\h{\b{u}}-u\big)=b^{-2},
\end{gather}
which can generate a rational part in the new solution $\b u$.

To f\/ind solutions from \eqref{BT-H1-R}, f\/irst, we introduce
\begin{subequations}\label{trans-BT-H1}
\begin{gather}
\t{u}-\b{u}=\frac{f\t{\b{f}}}{a\t{f}\b{f}},\label{trans-BT-H1-a}\\
\t{\b{u}}-u=\frac{\t{f}\b{f}}{af\t{\b{f}}},\label{trans-BT-H1-b}\\
\h{u}-\b{u}=\frac{f\h{\b{f}}}{b\h{f}\b{f}},\label{trans-BT-H1-c}\\
\h{\b{u}}-u=\frac{\h{f}\b{f}}{bf\h{\b{f}}},\label{trans-BT-H1-d}
\end{gather}
\end{subequations}
which provide a factorization of \eqref{BT-H1-R}. Such an assumption coincides with the previous results. In fact, suppose $V=\b f/f$, then from~\eqref{trans-BT-H1} we can f\/ind $\t u-\h u$ and $\th u-u$ agree with \eqref{BT-H1-mkdv} and $\t{\t u}-u$ and $\h{\h{u}}-u$ agree with~\eqref{u-V}. Then we introduce
\begin{gather}\label{trans-H1}
 u=x_{-1}-\frac{g}{f},\qquad \b u=x_{-1}-\frac{\b g}{\b f},
\end{gather}
by which we bilinearize \eqref{trans-BT-H1} as
\begin{subequations}\label{BT-H1-bil}
\begin{gather}
\b{g}\t{f}-\b{f}\t{g}+\frac{1}{a}(\b{f}\t{f}-\t{\b{f}}f)=0,\label{BT-H1-bil-a}\\
g\t{\b{f}}-f\t{\b{g}}-\frac{1}{a}(\b{f}\t{f}-\t{\b{f}}f)=0,\label{BT-H1-bil-b}\\
\b{g}\h{f}-\b{f}\h{g}+\frac{1}{b}(\b{f}\h{f}-\h{\b{f}}f)=0,\label{BT-H1-bil-c}\\
g\h{\b{f}}-f\h{\b{g}}-\frac{1}{b}(\b{f}\h{f}-\h{\b{f}}f)=0.\label{BT-H1-bil-d}
\end{gather}
\end{subequations}

Next, we introduce Casoratian forms for $f$, $\b f$, $g$ and $\b g$. Consider function
\begin{gather}
 \psi_i(n,m,l)= \psi_i^{+}(n,m,l) + \psi_i^{-}(n,m,l),\nonumber\\
\psi_i^{\pm}(n,m,l)=\varrho^{\pm}_i(1\pm s_i)^{l}(1\pm as_i)^n(1\pm bs_i)^m, \label{psi}
\end{gather}
where $\varrho^{\pm}_i$ and $s_i$ are nonzero constants\footnote{If $\varrho^{\pm}_i$ are independent on $s_i$, in practice in~\eqref{psi} we replace
$(1\pm s_i)^{l}$ with $(1\pm s_i)^{l+l_0}$ and suppose $l_0$ is either a large enough integer or a non-integer so that the derivative $\partial^{h}_{s_i}(1\pm s_i)^{l+l_0}|_{s_i=0}\neq 0$.}. This can be used to construct soliton solutions for H1 equation (cf.~\cite{HZ-JPA-2009})\footnote{One needs to use gauge property of bilinear H1 and make certain extension from $(\pm s_i)^{l+l_0}$ to $(1\pm s_i)^{l+l_0}$.}. To derive the rational solutions obtained in the previous section, we take
\begin{gather}\label{varrho}
 \varrho^{\pm}_i=\pm \frac{1}{2}\exp \left[{-\sum^{\infty}_{j=1}\frac{(\mp s_i )^{j}}{j}\gamma_j }\right]
\end{gather}
with arbitrary constant $\gamma_j$. Then we expand $\psi_i^{\pm}(n,m,l)$ as
\begin{gather}
\psi_i^{\pm}(n,m,l)=\pm \frac{1}{2}\sum^{\infty}_{h=0}\alpha^{\pm}_h s_i^h,\qquad \alpha^{\pm}_h=\pm\frac{2}{h!}\partial^{h}_{s_i}\psi_i^{\pm}|_{s_i=0}.
\label{alpha-psi}
\end{gather}

By noticing that
\begin{gather*}
\psi_i^{\pm}(n,m,l)=\pm \frac{1}{2}\exp \left[-\sum^{\infty}_{j=1}\frac{(\mp s_i)^{j}}{j}\c{x}_j \right],\qquad \c{x}_j=x_j+l,
\end{gather*}
where $x_j$ are exactly def\/ined as \eqref{xi}, all $\{\alpha^{\pm}_h\}$ can be expressed in terms of $\{x_j\}$. For $\{\alpha^{+}_h\}$ we have
\begin{gather}
\alpha^{\pm}_h\doteq \alpha^{\pm}_h(n,m,l)=(\mp1)^h\sum_{||\mu||=h}(-1)^{|\mu|}\frac{\c{\mathbf{x}}^{\mu}}{\mu!},\label{alpha+h}
\end{gather}
where
\begin{gather*}
\mu =(\mu_1,\mu_2,\dots),\qquad \mu_j\in \{0, 1, 2,\ldots\},\qquad ||\mu||=\sum^{\infty}_{j=1}j\mu_j, \\
|\mu|=\sum^{\infty}_{j=1}\mu_j, \qquad \mu ! =\mu_1!\cdot \mu_2!\cdots,\qquad
{\c{\mathbf{x}}}^{\mu} =\left(\frac{\c x_1}{1}\right)^{\mu_1}\left( \frac{\c x_2}{2}\right)^{\mu_2}\cdots.
\end{gather*}
The f\/irst few $\alpha^{+}_h$ are
\begin{gather*}
\alpha^+_0=1,\qquad \alpha^+_1=\c{x}_1,\qquad \alpha^+_2= \frac{1}{2}\big(\c{x}_1^2 - \c{x}_2\big),\qquad
\alpha^+_3=\frac{1}{6}\big(\c{x}_1^3 - 3 \c{x}_1 \c{x}_2 + 2 \c{x}_3\big),\\
\alpha^+_{4}=\frac{1}{24}\big(\c{x}_1^4-6\c{x}_1^2\c{x}_2+ 8 \c{x}_1\c{x}_3 +3 \c{x}_2^2 -6 \c{x}_4\big),\\
\alpha^+_5=\frac{1}{120}\big(\c{x}_1^5 - 10 \c{x}_2\c{x}_1^3 + 20\c{x}_3\c{x}_1^2 + 15 \c{x}_2^2\c{x}_1 - 30\c{x}_4\c{x}_1 - 20\c{x}_2\c{x}_3 + 24\c{x}_5\big).
\end{gather*}
Introduce a column vector
\begin{gather}
\alpha(n,m,l)=(\alpha_0, \alpha_1, \dots, \alpha_{N-1})^{\rm T},\qquad \alpha_j=\alpha_{2j+1}^{+}.\label{alpha}
\end{gather}

With $\alpha(n,m,l)$ as a basic column vector we introduce Casoratians w.r.t.\ shifts in $l$:
\begin{subequations}\label{fgfg-bar}
\begin{gather}
 f=\big|\h{N-1}\big|_{\rm R}=|\alpha(n,m,0),\alpha(n,m,1),\dots,\alpha(n,m,N-1)|, \label{fgfg-bar-f}\\
 \b f=\big|\h N\big|_{{\rm R}},\qquad g=\big|\h{N-2},N\big|_{\rm R}-Nf,\qquad \b g =\big|\h{N-1},N+1\big|_{{\rm R}}-(N+1)\b f.
\end{gather}
\end{subequations}
Some $f$, $g$ of low orders are
\begin{subequations}
\begin{gather}
f_{N=1}=x_1, \qquad g_{N=1}=1,\label{f1}\\
f_{N=2}=\frac{x_1^3-x_3}{3},\qquad g_{N=2}=x_1^2, \label{f2} \\
f_{N=3}=\frac{1}{45}x_1^6-\frac{1}{9}x_1^3x_3+\frac{1}{5}x_1x_5-\frac{1}{9}x_3^2,\qquad g_{N=3}=\frac{2}{15}x_1^5-\frac{1}{3}x_1^2x_3+\frac{1}{5}x_5.\label{f3}
\end{gather}\label{f-N}
\end{subequations}
Through \eqref{trans-H1}, $(f,g)$ with $N=1,2$ provide solutions \eqref{u3} and \eqref{u4} for~H1. For general~$N$, we have the following.

\begin{Theorem}\label{thm-5}The Casoratians \eqref{fgfg-bar} solve the bilinear BT~\eqref{BT-H1-bil} and \eqref{trans-H1} provides rational solutions to~{\rm H1}.
\end{Theorem}
Proof will be given in Appendix \ref{A-1}.

\begin{Remark} There is an alternative choice for the Casoratians \eqref{fgfg-bar}, which are given by just replacing the basic column vector $\alpha$ given in~\eqref{alpha} by
\begin{gather}
\beta(n,m,l)=(\beta_0, \beta_1, \dots, \beta_{N-1})^{\rm T}, \qquad \beta_j=\alpha_{2j}^{+},\label{beta}
\end{gather}
where $\alpha_{2j}^{+}$ are def\/ined in \eqref{alpha+h}, or equivalently,
\begin{gather*}
\beta_j=\frac{1}{(2j)!}\partial^{2j}_{s_i}\psi_i\big|_{s_i=0}
\end{gather*}
with
\begin{gather*} \varrho^{\pm}_i=\frac{1}{2}\exp \left[{-\sum^{\infty}_{j=1}\frac{(\mp s_i )^{j}}{j}\gamma_j }\right].\end{gather*}
\end{Remark}

\subsection{Casoratian solutions to (\ref{P-iterat})}\label{sec-5-3}

We can make use of the BT of H1 to obtain solutions to bilinear equation \eqref{P-iterat}. By the compatibility of \eqref{trans-BT-H1-a} and \eqref{trans-BT-H1-b}, i.e., $(E_n-E_N)(E_nE_N-1)u=(E_nE_N-1)(E_n-E_N)u$ where $E_N f=\b f$, we f\/ind
\begin{gather*}
 \frac{\t{\b{\b{f}}}f-\b{\b{f}}\t{f}}{\t{\b{f}} \,\b{f}} =E_n\left(\frac{\t{\b{\b{f}}}f-\b{\b{f}}\t{f}}{\t{\b{f}} \, \b{f}}\right).
\end{gather*}
Similarly,
\begin{gather*}
 \frac{\h{\b{\b{f}}}f-\b{\b{f}}\h{f}}{\h{\b{f}} \,\b{f}} =E_m \left(\frac{\h{\b{\b{f}}}f-\b{\b{f}}\h{f}}{\h{\b{f}}\, \b{f}}\right).
\end{gather*}
This means
\begin{gather}
\t{\b{\b{f}}}f-\b{\b{f}}\t{f}=\lambda_1(m,N)\t{\b{f}}\, \b{f},\qquad \h{\b{\b{f}}}f-\b{\b{f}}\h{f}=\lambda_2(n,N)\h{\b{f}}\,\b{f}. \label{f-lambda}
\end{gather}
Next we go to prove $\lambda_1(m,N)=a$ and $\lambda_2(n,N)=b$. Again, from \eqref{trans-BT-H1}, we can derive
\begin{gather*}
 \b{\b{u}}-u=\frac{\b{f}\t{f}}{a\t{\b{f}}f} -\frac{\b{f}\,\t{\b{\b{f}}}}{a\t{\b{f}}\,\b{\b{f}}} =\frac{\b{f}\h{f}}{b\h{\b{f}}f}-\frac{\b{f}\, \h{\b{\b{f}}}}{b\h{\b{f}}\,\b{\b{f}}}.
\end{gather*}
Using \eqref{f-lambda} to eliminate $\t{\b{\b{f}}}$ and $\h{\b{\b{f}}}$ from the above equation, we f\/ind
\begin{gather}
\b{\b{u}}-u=-\lambda_1(m,N)\frac{\b{f}^2}{af \b{\b{f}}}=-\lambda_2(n,N)\frac{\b{f}^2}{bf \b{\b{f}}},\label{lambda12}
\end{gather}
which means
\begin{gather*}
 a\lambda_2(n,N)=b\lambda_1(m,N),
\end{gather*}
and it then follows that both $\lambda_1$ and $\lambda_2$ must be $(n,m)$-independent. We assume
\begin{gather*}\gamma(N)=\lambda_1/a=\lambda_2/b,\end{gather*}
and then \eqref{f-lambda} yields
\begin{gather}
\gamma(N)=\frac{\t{\b{\b{f}}}f-\b{\b{f}}\t{f}}{a\t{\b{f}}\,\b{f}}=\frac{\h{\b{\b{f}}}f-\b{\b{f}}\h{f}}{b\h{\b{f}}\,\b{f}}. \label{f-gamma}
\end{gather}

To determine the value of $\gamma(N)$, we investigate properties of $f$ near the point $(n,m)=(0,0)$, which are presented through the following lemmas.

\begin{Lemma}\label{lem-f00} According to the definitions of $u$ in \eqref{trans-H1}, $f$ and $g$ in \eqref{fgfg-bar} and $\alpha^{\pm}_h$ in \eqref{alpha+h},
we find the value of $\alpha^{\pm}_h|_{(n,m)=(0,0)}$ is independent of $(a,b)$, and so are $f(0,0)$, $g(0,0)$ and $u(0,0)$. Then, from \eqref{lambda12} we find that~$\gamma(N)$ must be independent of $(a,b)$.
\end{Lemma}

\begin{Lemma}\label{lem-8}For same $N$, there exists relation
\begin{gather}
f_N(\alpha(n,m,l))= f_{N+1}(\beta(n,m,l)),\label{f-alpha-beta}
\end{gather}
where $\alpha(n,m,l)$ and $\beta(n,m,l)$ are respectively $N$-th order and $(N+1)$-th order column vectors defined as~\eqref{alpha} and~\eqref{beta}. Here and below $f_N(\psi)$ stands for a $N$-th order Casoratian $|\h{N-1}|$ composed by a $N$-th order basic column vector~$\psi$.
\end{Lemma}

\begin{proof} First, noticing that relation
\begin{gather*}\psi_i^{\pm}(n,m,l+1)-\psi_i^{\pm}(n,m,l)=\pm s_i \psi_i^{\pm}(n,m,l),\end{gather*}
from the def\/inition of $\alpha^+_h$ in~\eqref{alpha-psi}, we immediately get
\begin{gather}\label{alpha-h+1}
 \alpha^+_h(n,m,l+1)- \alpha^+_h(n,m,l)= \alpha^+_{h-1}(n,m,l), \qquad h\geq1,
\end{gather}
from which, taking $h=2j$, we reach
\begin{gather*}
 \beta_j(n,m,l+1)-\beta_j(n,m,l)=\alpha_{j-1}(n,m,l), \qquad j\geq 1.
\end{gather*}
It then follows that
\begin{gather*}
\beta(n,m,l+1)-\beta(n,m,l)=\left(\begin{matrix}0\\ \alpha(n,m,l) \end{matrix}\right),
\end{gather*}
where $\alpha(n,m,l)$ and $\beta(n,m,l)$ are respectively $N$-th order and $(N+1)$-th order column vectors def\/ined as~\eqref{alpha} and~\eqref{beta}. This immediately leads to the relation~\eqref{f-alpha-beta}.
\end{proof}

\begin{Lemma}\label{lem-9} For Casoratian $f_N(\alpha(n,m,l))$, the relation
\begin{gather*}f_N(\alpha(1,0,l))=a^N f_{N-1}(\alpha(0,0,l))+ O\big(a^{N-1}\big)\end{gather*}
holds.
\end{Lemma}
\begin{proof}
\begin{gather*}
f_N(\alpha(1,0,l))=f_N(a\beta(0,0,l)+\alpha(0,0,l))\\
\hphantom{f_N(\alpha(1,0,l))}{} = a^N f_N(\beta(0,0,l))+ O\big(a^{N-1}\big) = a^N f_{N-1}(\alpha(0,0,l))+ O\big(a^{N-1}\big),
\end{gather*}
where we have made use of relation \eqref{f-alpha-beta}.
\end{proof}

With this lemma, for $f=f_N(\alpha(n,m,l))$ in \eqref{f-gamma}, we have
\begin{gather*}
\t f|_{n=m=0}=a^N \ub f|_{n=m=0} + O\big(a^{N-1}\big), \qquad \t{\b f}|_{n=m=0}=a^{N+1} f|_{n=m=0} + O\big(a^{N}\big),\\
\t{\b{\b f}}|_{n=m=0}=a^{N+2} \b f|_{n=m=0} + O\big(a^{N+1}\big).
\end{gather*}
Then, since $\gamma(N)$ is independent of $a$, from \eqref{f-gamma} we arrive at
\begin{gather*}\gamma(N)= \lim_{a\to \infty} \frac{\t{\b{\b{f}}}f-\b{\b{f}}\t{f}}{a\t{\b{f}}\,\b{f}}\Biggr|_{n=m=0}=1.\end{gather*}

We can sum up this subsection with the following theorem.
\begin{Theorem}\label{thm-6}
The Casoratian $f=f_N(\alpha(n,m,l))$ solves bilinear equation set
\begin{subequations}\label{f-itera}
\begin{gather}
\t{\b{\b{f}}}f-\b{\b{f}}\t{f}=a\t{\b{f}}\,\b{f}, \label{f-itera-a}\\
 \h{\b{\b{f}}}f-\b{\b{f}}\h{f}=b\h{\b{f}}\,\b{f}. \label{f-itera-b}
\end{gather}
\end{subequations}
$P=f_N(\alpha(n,m,l))$ provides a Casoratian form of solution to \eqref{P-iterat}. By defining
\begin{gather}
f_{-N}=(-1)^{[\frac{N}{2}]}f_{N-1},\qquad f_0=1, \label{f-N-minus}
\end{gather}
one can consistently extend \eqref{f-itera} to $N\in \mathbb{Z}$, which coincides with~\eqref{P-N-minus}.
\end{Theorem}

\subsection{Casoratian rational solutions to H2 and a sum-up}

We can derive Casoratian rational solutions for H2 through non-auto BT \eqref{BT-H1-H2}, in which we suppose
\begin{gather}\label{trans-H1H2}
 u=x_{-1}-\frac{g}{f}, \qquad v=x_{-1}^2-2(x_{-1}+N)\frac{g}{f}+\frac{h}{f}-N^2.
\end{gather}
Then BT \eqref{BT-H1-H2} is bilinearized as
\begin{subequations}\label{BT-H1H2-bil}
\begin{gather}
 f\t{h}+\t{f}h+2\big(a^{-1}-N\big)g\t{f}-2\big(a^{-1}+N\big)\t{g}f-2g\t{g}-2N^2f\t{f}=0,\label{BT-H1H2-bila}\\
 f\h{h}+\h{f}h+2\big(b^{-1}-N\big)g\h{f}-2\big(b^{-1}+N\big)\h{g}f-2g\h{g}-2N^2f\h{f}=0.\label{BT-H1H2-bilb}
\end{gather}
\end{subequations}
Based on the bilinear form we have

\begin{Theorem}\label{thm-7}The Casoratians
\begin{gather}
f=\big|\h{N-1}\big|_{\rm R}, \qquad g=\big|\h{N-2},N\big|_{\rm R}-Nf, \nonumber\\
 h=\big|\h{N-2},N+1\big|_{\rm R}+\big|\h{N-3},N-1,N\big|_{\rm R}\label{fgh}
\end{gather}
solve the bilinear BT system \eqref{BT-H1H2-bil}, in which the basic Casoratian column vector~$\alpha$ is given by~\eqref{alpha}. Consequently, \eqref{trans-H1H2} provides rational solutions to {\rm H1} and~{\rm H2}.
\end{Theorem}
Proof will be given in Appendix \ref{A-2}.

Besides \eqref{f-N}, some $h$ of low orders are{\samepage
\begin{gather*}
h_{N=1}=x_1+2,\qquad
h_{N=2}=\frac{4}{3}\big(x_1^3-x_3\big)+4x_1^2+2x_1, \\
h_{N=3}=\frac{1}{5}x_1^6-x_1^3x_3+\frac{9}{5}x_1x_5-x_3^2+\frac{4}{5}x_1^5-2x_1^2x_3+\frac{6}{5}x_5+\frac{2}{3}x_1^4-\frac{2}{3}x_1x_3,
\end{gather*}
where all $\gamma_i=l_i$ in $x_i$.}

So far we have obtained Casoratian expressions for the rational solutions of Q1(0), lpmKdV, Q1($\delta$), H3($\delta$), H1 and H2. Noting that all these solutions are related to the rational solutions~$V_N$ of the lpmKdV equation, it is necessary to express all these obtained solutions through the Casoratians with a unif\/ied $N$. We collect them in the following theorem.

\begin{Theorem}\label{thm-X}Suppose that
\begin{gather}
 f\doteq f_N =\big|\h{N-1}\big|_{\rm R}, \qquad g\doteq g_N =\big|\h{N-2},N\big|_{\rm R}-Nf, \nonumber\\
 h\doteq h_N =\big|\h{N-2},N+1\big|_{\rm R}+\big|\h{N-3},N-1,N\big|_{\rm R},\label{fgh-X}
\end{gather}
and denote $\b f =f_{N+1}$ and $\ub f=f_{N-1}$. Then the rational solutions for {\rm Q1(0)}, {\rm lpmKdV}, {\rm Q1($\delta$)}, {\rm H3($\delta$)}, {\rm H1} and {\rm H2} are respectively
\begin{subequations}\label{rs-X}
\begin{alignat}{3}
& {\rm Q1(0)}\colon \quad && v_{N+2}=\frac{\b{\b f}}{ f},& \\
& {\rm lpmKdV}\colon \quad && V_{N+2}=\frac{\b f}{ f},& \\
& \mathrm{Q1(\delta)}\colon \quad && u_{N+2}=\frac{\b{\b f} + \delta^2 \ub{\ub f}}{f}, & \label{rs-X-Q1del}\\
& \mathrm{H3(\delta)}\colon \quad && Z_{N+2} = (-1)^{\frac{n+m}{2}+\frac{1}{4}}\frac{\b f +(-1)^{n+m}\delta \ub{f}}{f},& \\
& {\rm H1}\colon \quad && u_{N+2}= x_{-1}-\frac{g}{f},& \\
& {\rm H2}\colon \quad &&v_{N+2}= x_{-1}^2-2(x_{-1}+N)\frac{g}{f}+\frac{h}{f}-N^2.&
\end{alignat}
\end{subequations}
\end{Theorem}

\subsection{Rational solutions to Q2}

Now we come to the f\/inal equation, Q2. We start from the non-auto BT \eqref{BT-Q1Q2} in which we take parametrization~\eqref{pq-ab-Q1} and~$u$ to be~\eqref{rs-X-Q1del} which is a solution of~Q1($\delta$). Introduce auxiliary function
\begin{gather*}w=y+\frac{u^2}{\delta^2}\end{gather*}
by which the BT \eqref{BT-Q1Q2} yields
\begin{gather}
\t{y}=\frac{\t{u}-u-\delta p}{\t{u}-u+\delta p}y+\frac{1}{\delta^2}(u+\delta p-\t{u})(u+\delta p+\t{u}),\nonumber\\
\h{y}=\frac{\h{u}-u-\delta q}{\h{u}-u+\delta q}y+\frac{1}{\delta^2}(u+\delta q-\h{u})(u+\delta q+\h{u}).\label{BT-Q1Q2-y}
\end{gather}
Then, making use of the relation \eqref{f-itera}, from \eqref{rs-X-Q1del} we can f\/ind
\begin{gather*}
\t{u}_{N+2}-u_{N+2}+\delta p=a\frac{(\b{f}+\delta \ub{f})\big(\t{\b{f}}-\delta \t{\ub{f}}\big)}{f\t{f}},\\
\t{u}_{N+2}-u_{N+2}-\delta p=a\frac{(\b{f}-\delta \ub{f})\big(\t{\b{f}}+\delta \t{\ub{f}}\big)}{f\t{f}}.
\end{gather*}
On the basis of the above relations together with their $(q, \h{~~})$ version, and introducing
\begin{gather*}\theta_{N+2}=y_{N+2}\frac{\b{f}-\delta\ub{f} ~}{\b{f}+\delta\ub{f}~},\end{gather*}
we then reduce \eqref{BT-Q1Q2-y} to
\begin{gather}
\theta_{N+2}-\t{\theta}_{N+2}=\frac{a(\b{f}-\delta\ub{f})\big(\t{\b{f}}-\delta\t{\ub{f}}\big)}{\delta^2 f\t{f}}(u_{N+2}+\t{u}_{N+2}+\delta p),\nonumber\\
\theta_{N+2}-\h{\theta}_{N+2}=\frac{b(\b{f}-\delta\ub{f})\big(\h{\b{f}}-\delta\h{\ub{f}}\big)}{\delta^2 f\h{f}}(u_{N+2}+\h{u}_{N+2}+\delta q).\label{theta-Q2}
\end{gather}
To solve this system we expand
\begin{gather*}\theta_{N+2}=\sum_{i=-2}^{2}\theta_{N+2}^{(i)}\delta^i.\end{gather*}
It then follows from \eqref{theta-Q2} that
\begin{subequations}
\begin{gather}
\theta_{N+2}^{(-2)}-\t{\theta}_{N+2}^{(-2)}=\frac{a\t{\b{f}}\,\b{f}}{f^2\t{f}^2}\big(\t{\b{\b{f}}}f+\t{f}\,\b{\b{f}}\big), \\
\theta_{N+2}^{(-1)}-\t{\theta}_{N+2}^{(-1)}=\frac{-2a}{f^2\t{f}^2}\big(\b{f}\,\t{\ub{f}}\,\t{\b{\b{f}}}\,f+\t{\b{f}}\ub{f}\t{f}\,\b{\b{f}}\big), \\
\theta_{N+2}^{(0)}-\t{\theta}_{N+2}^{(0)}=\frac{a^2}{f^2\t f^2}\big(\t{\ub{f}}{}^2 \b f^2-\ub{f}^2 \tb f {}^2\big)
+\frac{2\big(\t{\b{\b{f}}}\ub{\ub{f}}-\t{\ub{\ub{f}}}\,\b{\b{f}}\big)}{f\t{f}}, \label{theta-0}\\
\theta_{N+2}^{(1)}-\t{\theta}_{N+2}^{(1)}=\frac{-2a}{f^2\t{f}^2}\big(\ub{f}\,\t{\ub{\ub{f}}}\t{\b{f}}f +\t{\ub{f}}\,\ub{\ub{f}}\t{f}\,\b{f}\big),\\
\theta_{N+2}^{(2)}-\t{\theta}_{N+2}^{(2)}=\frac{a\t{\ub{f}}\,\ub{f}}{f^2\t{f}^2}\big(\t{\ub{\ub{f}}}f+\t{f}\ub{\ub{f}}\big),
\end{gather}
\end{subequations}
among which, except $\theta_{N+2}^{(0)}$, we f\/ind explicit expressions for $\theta_{N+2}^{(i)}$ in terms of $f$:
\begin{gather}
 \theta_{N+2}^{(-1)}=-\frac{\b{\b{f}}{}^2}{f^2}, \qquad \theta_{N+2}^{(-1)}=\frac{2\b{\b{f}}{}^2\ub{f}+2f^2\b{\b{\b{f}}}}{ f^2\b{f}},\nonumber\\
 \theta_{N+2}^{(1)}=-\frac{2\ub{\ub{f}}^2\b{f}+2 f^2\ub{\ub{\ub{f}}}~}{ f^2\ub{f}},\qquad
 \theta_{N+2}^{(2)}=\frac{\ub{\ub{f}}^2}{f^2}.
 \label{theta-f}
\end{gather}
For $\theta_{N+2}^{(0)}$ which is determined by \eqref{theta-0}, the simplest two items are
\begin{gather*}
\theta^{(0)}_{2}=\frac{1}{3}x_1^4+\frac{2}{3}x_1x_3,\qquad \theta^{(0)}_{3}=-\frac{1}{15}x_1^4+\frac{2}{3}x_1x_3+\frac{2x_5}{5x_1}.
\end{gather*}
However, so far we do not f\/ind an explicit expression for $\theta_{N+2}^{(0)}$ in terms of $f$ and other auxiliary functions.

As a conclusion of rational solutions of Q2, we give the following theorem.

\begin{Theorem}\label{thm-8}
Suppose that $f=f_N$ is defined as in \eqref{fgh-X}. Our construction provides rational solutions of {\rm Q2} in the following form
\begin{gather}\label{81}
 w_{N+2}=\frac{u^2_{N+2}}{\delta^2}+\frac{\b{f}+\delta\ub{f}}{\b{f}-\delta\ub{f}}\!
 \left(\frac{-\b{\b{f}}^2}{\delta^2 f^2}+\frac{2\b{\b{f}}{}^2\ub{f}+2f^2\b{\b{\b{f}}}}{\delta f^2\b{f}}
 +\theta_{N+2}^{(0)} -\frac{2\delta \ub{\ub{f}}^2\b{f}+2\delta f^2 \ub{\ub{\ub{f}}}}{ f^2\ub{f}}
 +\frac{\delta^2\ub{\ub{f}}^2}{f^2}\right)\!,\!\!\!
\end{gather}
where $u_{N+2}$ is given by \eqref{rs-X-Q1del} and $\theta_{N+2}^{(0)}$ is determined by~\eqref{theta-0}.
\end{Theorem}

We note that it might be not suf\/f\/icient to call \eqref{81} a rational solution for arbitrary $N$, because for this moment we do not have a general solution form (like \eqref{theta-f}) for $\theta_{N+2}^{(0)}$.

\section{Conclusions}\label{sec-6}

In the paper we have derived rational solutions for the lpmKdV equation and some lattice equations in the ABS list. We make use of lpmKdV-Q1(0) consistent triplet to construct their rational solutions iteratively. This then becomes a starting point and through the route in Fig.~\ref{fig1} to generate solutions for other equations. All these rational solutions are related to a~unif\/ied $\tau$ function in Casoratian form, $f(\alpha)=|\h{N-1}|$, which obeys the bilinear superposition formula~\eqref{f-itera}.

There are several interesting points we would like to remark. First, formula~\eqref{u-N} reveals an explicit relation between certain solutions of Q1($\delta$) and Q1(0). This formula holds not only for rational solutions but also for solitons. Once we obtain $v_{N-2}$ and $v_N$ from~\eqref{iterat}, formula~\eqref{u-N} gives a solution $u_N$ to Q1($\delta$), and these solutions provide a solution sequence for the chain~\eqref{chain} which is based on BT \eqref{BT:Q1Q1}, i.e.,~\eqref{BT-Q1Q10-N}. The second thing is about bilinear superposition formula~\eqref{f-itera} or~\eqref{P-iterat}. Casoratian $f$ with $\psi_i$ \eqref{psi} as a basic entry is also a~solution of bilinear equation
\begin{gather}
(a+b)\h{\ut f}\t f+ (b-a) \ut f \h{\t f} =2b f \h f,
\label{b-1}
\end{gather}
as well as its dual version by switching $(a,\t{~~})$ and $(b,\h{~~})$. \eqref{b-1} can be considered as a bilinear form of Hirota's discrete KdV equation (see \cite{Hir-JPSJ-1977-I} and \cite[Section~8.4.1]{HJN-book-2016}). It was also derived from the Cauchy matrix approach as a bilinear form that is related to~H1 (see \cite[Section~9.4.3]{HJN-book-2016}). It is also well known that \eqref{b-1} can be derived as a reduction of the Hirota--Miwa equation, of which some rational solutions were derived from several dif\/ferent ways and reductions of few cases was already considered \cite{GRPSW-JPA-2007, MKNO-PLA-1997}.
Here we can consider \eqref{f-itera} as a bilinear superposition formula of~\eqref{b-1} for rational solutions. Since \eqref{f-itera} holds for all $N\in \mathbb{Z}$, it might be possible to connect \eqref{f-itera} with some 3D lattice equations. Finally, let us go back to $x_i$ def\/ined in~\eqref{xi}. It is interesting that all the $\{\alpha_N\}$ can be expressed in terms of~$x_i$. Recalling Lemma~\ref{lem-3} in which $v_N$ can be positive in the f\/irst quadrant $\{n\geq 0,\, m\geq 0\}$ if we take $v_N(0,0)>0$ which can be done by suitably choosing value for $\gamma_{2N-1}$ (see~\eqref{v123} as examples), we can make use of the relation between $v_N$ and $f$ to formulate a mechanism for choosing~$\gamma_j$ so that~$f$ is nonzero in the f\/irst quadrant. This will be done in Appendix~\ref{A-3}.

At the end of the paper we would like to make a comparison for the rational solutions and their derivation between the present paper and~\cite{SZ-SIGMA-2011}. In this paper the construction of rational solutions is based on iteration of a chain of transformations, and the unif\/ied $\tau$ function $f(\alpha)=|\h{N-1}|$ is proved to satisfy the bilinear superposition formula~\eqref{f-itera}. In~\cite{SZ-SIGMA-2011}, rational solutions (most of them with exponential background) for H3$(\delta)$ and Q1($\delta$) are obtained via a limiting procedure from soliton solutions in Casoratian expression. The method used in~\cite{SZ-SIGMA-2011} can be extended to H1 and H2 (by selecting~\eqref{psi} as a basic Casoratian entry) and the results will be the same as the present paper. However, for H3$(\delta)$ and Q1($\delta$) it is obvious that our construction, which brings pure rational solutions, allows reduction $\delta=0$ and relies only on a unif\/ied $\tau$ function, has more advantage than the limiting procedure used in~\cite{SZ-SIGMA-2011}. It is hard to say what is the reason of this dif\/ference, but a fact is all the BTs we used in our paper are only parametrically related to spacing parameters~$a$,~$b$ without any extra parameters for solitons. These BTs are natural for generating rational solutions.

\appendix

\section{Proof of Theorem \ref{thm-5} for H1}\label{A-1}

Here we prove Theorem \ref{thm-5} which gives Casoratian form of rational solutions of~H1.

First, we prove \eqref{BT-H1-bil-a}. Noticing that $\psi_i$ def\/ined in~\eqref{psi} satisf\/ies shift relation
\begin{gather*}
 \psi_i(l)-a\wideutilde{\psi}_i(l+1)=(1-a)\wideutilde{\psi}_i(l)
\end{gather*}
and $\psi_i$ (with $\varrho^{\pm}_i$ \eqref{varrho}) and $\alpha_j$ def\/ined by \eqref{alpha} actually obey the relation
\begin{gather}
\psi_i(l)=\sum^{\infty}_{j=0}\alpha_j(l) s_i^{2j+1}\label{alpha-phi}
\end{gather}
we have
\begin{gather}
 \alpha_i(l)-a\ut{\alpha}_i(l+1)=(1-a)\ut{\alpha}_i(l).
\label{alpha-shift}
\end{gather}
With such a shift relation and using the technique in \cite{HZ-JPA-2009}, for the Casoratians in \eqref{fgfg-bar} we f\/ind
\begin{subequations}\label{fg-shiht-1}
\begin{gather}
(1-a)^{N-1}\ut{f}=\big|\h{N-2},\ut{\alpha}(N-1)\big|,\\
-a(1-a)^{N-1}\b{\ut{f}}=\big|\h{N-1},\ut{\alpha}(N-1)\big|,\\
-a(1-a)^{N-1}\b{\ut{g}}=\big|\h{N-2},N,\ut{\alpha}(N-1)\big|+(1-a)^{N-1}(1+aN)\b{\ut{f}},
\end{gather}
\end{subequations}
where we have neglected subscript ``$R$'' without making any confusion.

Substituting \eqref{fgfg-bar} and \eqref{fg-shiht-1} into the downtilde-shifted \eqref{BT-H1-bil-a}, for the l.h.s.\ we reach
\begin{gather}
 \big|\widehat{N}\big|\big|\widehat{N-2},\ut{\alpha}(N-1)\big|\!-\!\big|\widehat{N-1},\ut{\alpha}(N-1)\big| \big|\widehat{N-2},N\big|\! + \!\big|\widehat{N-1}\big|\big|\widehat{N-2},N,\ut{\alpha}(N-1)\big|,\!\!\!\!
\label{BT-H1-bil-a=0}
\end{gather}
which is zero in light of Lemma \ref{L:lap}. In fact, we can replace the $N$-th order vector $\alpha$ with $(N+1)$-th order one, introduce an auxiliary $(N+1)$-th order column vector $e_{N+1}=(0,0,\dots,0,1)^{\rm T}$, and rewrite
\begin{gather*}f=\big|\h{N-1},e_{N+1}\big|,\qquad g=\big|\h{N-2},N,e_{N+1}\big|,\\ \big|\h{N-2},\ut{\alpha}(N-1)\big|=\big|\h{N-2},\ut{\alpha}(N-1),e_{N+1}\big|;
\end{gather*}
then after taking $\mathbf{B}=(\h{N-2})$, $\mathbf{a}=\alpha(N-1)$, $\mathbf{b}=e_{N+1}$, $\mathbf{c}=\alpha(N)$, $\mathbf{d}=\ut \alpha(N-1)$,
\eqref{BT-H1-bil-a=0}~vanishes due to Lemma~\ref{L:lap}.

Next, to prove \eqref{BT-H1-bil-b} we consider Casoratians $f$ and $g$ composed by $\phi(l)=(\phi_1,\phi_2,\dots,\phi_N)^{\rm T}$ where
\begin{gather*}
 \phi_i(n,m,l)=\varrho^{+}_i(1+s_i)^{l}(1-as_i)^{-n}(1+bs_i)^m+\varrho^{-}_i(1-s_i)^{l}(1+as_i)^{-n}(1-bs_i)^m,
\end{gather*}
which satisf\/ies
\begin{gather}
 \phi_i(l)+a\t{\phi}_i(l+1)=(1+a)\t{\phi}_i(l).
\label{phi-shift}
\end{gather}
Introduce vector
\begin{gather*}
\omega(l)=(\omega_1(l),\omega_2(l),\dots,\omega_N(l))^{\rm T},\qquad \omega_j=\frac{1}{(2j+1)!}\partial^{2j+1}_{s_i}\phi_i|_{s_i=0}.
\end{gather*}
Noticing the expression \eqref{alpha-phi} for $\alpha_j(l)$ and relation
\begin{gather*}\phi_i=\frac{1}{(1-a^2s^2_i)^n}\psi_i\end{gather*}
where we have taken $\varrho^{\pm}_i$ def\/ined as \eqref{varrho}, we f\/ind
\begin{gather}
\omega=A\alpha,\label{ome-alp}
\end{gather}
where $A=(a_{ij})_{N\times N}$ is a lower triangular Toeplitz matrix
def\/ined by
\begin{gather*}
a_{ij}= \begin{cases}
 0, & i<j,\\
\displaystyle \frac{\partial^{2(i-j)}_{s_i}}{[2(i-j)]!}\frac{1}{(1-a^2s^2_i)^n}\Bigr|_{s_i=0}, & i\geq j.
 \end{cases}
\end{gather*}
Noticing the relation \eqref{ome-alp} and $|A|=1$, we have
\begin{gather*}
f(\omega(l))=|A|f(\alpha(l))=f(\alpha(l)),\qquad g(\omega(l))=g(\alpha(l)).
\end{gather*}
Besides, $\omega_i$ obeys the same shift relation as \eqref{phi-shift}, i.e.,
\begin{gather}
 \omega_i(l)+a\t{\omega}_i(l+1)=(1+a)\t{\omega}_i(l),\label{omega-shift}
\end{gather}
which leads to
\begin{gather*}
(1+a)^{N-1}\t{f}(\omega(l)) =\big|\h{N-2},\t{\omega}(N-1)\big|,\\
a(1+a)^{N-1}\t{\b{f}}(\omega(l))=\big|\h{N-1},\t{\omega}(N-1)\big|,\\
a(1+a)^{N-1}\t{\b{g}}(\omega(l))=\big|\h{N-2},N,\t{\omega}(N-1)\big|-(1+a)^{N-1}(Na-1)\t{\b{f}}.
\end{gather*}
Then one can f\/ind the l.h.s.\ of \eqref{BT-H1-bil-b} yields
\begin{gather*}
\big|\widehat{N}\big|\big|\h{N-2},\t{\omega}(N-1)\big|-\big|\h{N-1},\t{\omega}(N-1)\big| \big|\h{N-2},N\big|
 + \big|\h{N-1}\big|\big|\h{N-2},N,\t{\omega}(N-1)\big|,
\end{gather*}
which vanishes as \eqref{BT-H1-bil-a=0}.

\eqref{BT-H1-bil-c} and \eqref{BT-H1-bil-d} can be proved similarly.

\section{Proof of Theorem \ref{thm-7} for H2}\label{A-2}

To prove Theorem \ref{thm-7}, we rewrite
\begin{gather*}
 h=s+t,\qquad s=\big|\h{N-2},N+1\big|_{\rm R},\qquad t=\big|\h{N-3},N-1,N\big|_{\rm R}.
\end{gather*}
With the relation \eqref{alpha-shift} and using the technique in \cite{HZ-JPA-2009}, for the Casoratians \eqref{fgh} we have
\begin{gather*}
 a(1-a)^{N-2}\big[\ut{s}+(a^{-1}-1)(\ut{g}+N\ut{f})\big]=-\big|\h{N-3},N,\ut{\alpha}(N-2)\big|,\\
 a(1-a)^{N-2}\big[\ut{g}+(a^{-1}+N-1)\ut{f}\big]=-\big|\h{N-3},N-1,\ut{\alpha}(N-2)\big|,\\
 a(1-a)^{N-2}\ut{f}=-\big|\h{N-3}, N-2,\ut{\alpha}(N-2)\big|.
\end{gather*}
Again, here and after we drop of\/f subscript ``$R$'' without making any confusion. Then we f\/ind that
\begin{gather}
 a(1-a)^{N-2}\big\{f\big[\ut{s}+(a^{-1}-1)(\ut{g}+N\ut{f})\big]-(g+Nf)\big[\ut{g}+(N+a^{-1}-1)\ut{f}\big]+\ut{f}t\big\} \nonumber\\
\qquad{} =-\big|\h{N-1}\big|\big|\h{N-3},N,\ut{\alpha}(N-2)\big|+\big|\h{N-2},N\big|\big|\h{N-3},N-1,\ut{\alpha}(N-2)\big| \nonumber\\
\qquad\quad{}-\big|\h{N-3}, N-2,\ut{\alpha}(N-2)\big|\big|\h{N-3},N-1,N\big| = 0.\label{fghbil1}
\end{gather}
Since
\begin{gather*}
f(\omega(l))=f(\alpha(l)),\qquad g(\omega(l))=g(\alpha(l)),\qquad s(\omega(l))=s(\alpha(l)),\qquad t(\omega(l))=t(\alpha(l)),
\end{gather*}
using \eqref{omega-shift} one has
\begin{gather*}
 a(1+a)^{N-2}\big[\t{s}(\omega)-\big(a^{-1}+1\big)(\t{g}(\omega)+N\t{f}(\omega))\big]=\big|\h{N-3},N, \t{\omega}(N-2)\big|,\\
 a(1+a)^{N-2}\big[\t{g}(\omega)-\big(a^{-1}-N+1\big)\t{f}(\omega)\big]=\big|\h{N-3},N-1, \t{\omega}(N-2)\big|,\\
 a(1+a)^{N-2}\t{f}(\omega)=\big|\h{N-2},\t{\omega}(N-2)\big|.
\end{gather*}
Consequently it reaches
\begin{gather}
 a(1+a)^{N-2}\big\{f\big[ \t{s}-(a^{-1}+1)(\t{g}+N\t{f})\big]-(g+Nf)\big[\t{g}-(a^{-1}-N+1)\t{f}\big]+\t{f}t\big\} \nonumber\\
\qquad {} =\big|\h{N-1}\big|\big|\h{N-3},N, \t{\omega}(N-2)\big|-\big|\h{N-2},N\big|\big|\h{N-3},N-1, \t{\omega}(N-2)\big| \nonumber\\
\qquad\quad{} +\big|\h{N-2},\t{\omega}(N-2)\big|\big|\h{N-3},N-1,N\big| = 0. \label{fghbil2}
\end{gather}
Then, adding \eqref{fghbil2} and the uptilde-shifted \eqref{fghbil1} yields \eqref{BT-H1H2-bila}. The other equation in \eqref{BT-H1H2-bil} can be proved similarly.

\section[Property of $f$]{Property of $\boldsymbol{f}$}\label{A-3}

In the following we take a close look at Casoratian $f_N$ def\/ined by \eqref{fgfg-bar-f}, i.e.,
\begin{gather}
f_N= \big|\h{N-1}\big| =|\alpha(n,m,0),\alpha(n,m,1),\dots,\alpha(n,m,N-1)|,\label{fN}
\end{gather}
where \looseness=-1 $\alpha$ is given by \eqref{alpha}. Due to relation \eqref{f-N-minus}, we only consider the case $N\in \mathbb{Z}^+$. To investigate properties of $f_N$, we introduce ``degree'' for a polynomial. For a monomial $\prod_{i\geq 1} x_i^{k_i}$ where~$x_i$ is def\/ined in~\eqref{xi}, we assign it a degree $\sum_{i\geq 1} ik_i$ and denote this number by $\mathcal{D}\big[\prod_{i\geq 1} x_i^{k_i}\big]$. A~polynomial $P=P[\{x_i\}]$ in which each monomial has same degree $d$ is called homogeneous and its degree is denoted by $\mathcal{D}[P]=d$. Under this def\/inition, for the $f_N$ given in~\eqref{f-N}, they are all homogeneous and their degrees are $\mathcal{D}[f_1]=0$, $\mathcal{D}[f_2]=3$, $\mathcal{D}[f_3]=6$. In particular, we have

\begin{Lemma}\label{lem-C-1}For $\alpha^{+}_h(n,m,l)$ defined in \eqref{alpha+h}, $\alpha^{+}_h(n,m,0)$ is homogeneous with degree
\begin{gather*}
\mathcal{D}[\alpha^{+}_h(n,m,0)]=h.
\end{gather*}
\end{Lemma}

Now we come to investigate properties of $f_N$.
\begin{Theorem}\label{thm-C-1} Casoratian $f_N$ \label{fNc} in which $\alpha$ is given by \eqref{alpha} has the following properties:
\begin{enumerate}\itemsep=0pt
\item[$(i)$] $f_N$ is homogeneous with degree $\mathcal{D}[f_N]=\frac{N(N+1)}{2}$;
\item[$(ii)$] $f_N$ depends only on $\{x_1, x_3, \dots, x_{2N-1}\}$;
\item[$(iii)$] $f_N(n,m)$ is positive in the first quadrant $\{n\geq 0, \, m\geq 0\}$ provided $a>0$, $b>0$ and $f_N(0,0)>0$;
\item[$(iv)$] in construction, $f_N(0,0)>0$ is guaranteed by successively choosing
\begin{gather*}
(-1)^{N+1}\gamma_{2N-1}>-\frac{(2N-1)f_{N}(0,0)|_{\gamma_{2N-1}=0}}{f_{N-2}(0,0)}.
\end{gather*}
\end{enumerate}
\end{Theorem}

\begin{proof} We prove the items of the theorem one by one.

(i) Consider Casoratian $f_N$ \eqref{fN} in which $\alpha$ is given by \eqref{alpha}. $f_N$ is written as $|(f_{ij})_{N\times N}|$ where $f_{ij}= \alpha^{+}_{2i-1}(n,m,j-1)$. Noting that it is not $\alpha^{+}_h(n,m,l)$ but $\alpha^{+}_h(n,m,0)$ that is homogeneous, in the following we make use of shift relation \eqref{alpha-h+1} to rewrite $f_N$ in terms of $\{\alpha^{+}_h(n,m,0)\}$. To do that, f\/irst from \eqref{alpha-h+1} we have{\samepage
\begin{subequations}\label{alp-bet}
\begin{gather}
 \alpha(l+1)-\alpha(l)=\beta(l),\label{alp-bet-1}\\
 \beta(l+1)-\beta(l)=\Lambda \alpha(l),\label{alp-bet-2}
\end{gather}
\end{subequations}}

\noindent
where
\begin{gather*}\Lambda = (\delta_{i,j+1})_{N\times N},\qquad \delta_{i,j}= \begin{cases}1, & i=j,\\ 0, & i\neq j.\end{cases}\end{gather*}
Here in \eqref{alp-bet} we have omitted $n$, $m$ in $\alpha(n,m,l)$ and $\beta(n,m,l)$ for convenience. Successively making use of \eqref{alp-bet-1} we can rewrite \eqref{fN} from the last to the f\/irst column, and we have
\begin{gather*}f_N=|\alpha(0), \beta(0), \beta(1), \beta(2), \dots,\beta(N-3), \beta(N-2)|.\end{gather*}
Then, employing \eqref{alp-bet-2} and in a similar manner we have
\begin{gather*}f_N=|\alpha(0), \beta(0), \Lambda \alpha(0), \Lambda \alpha(1), \dots, \Lambda \alpha(N-4), \Lambda \alpha(N-3)|.\end{gather*}
This procedure can be continued until we arrive at
\begin{gather}
f_N=\big|\alpha(0), \beta(0), \Lambda \alpha(0), \Lambda \beta(0), \Lambda^2 \alpha(0), \Lambda^2 \beta(0),\dots, \Lambda^{[(N-1)/2]} \sigma(0)\big|,
\label{fN'}
\end{gather}
where
\begin{gather*}\sigma(0)= \begin{cases}
 \alpha(0), & N \ {\rm odd},\\
 \beta(0), & N \ {\rm even}.
\end{cases}
\end{gather*}
\eqref{fN'} can be denoted by $f_N=|(f^{'}_{ij})_{N\times N}|$ where
\begin{gather*}f^{'}_{ij}= \begin{cases}
 \alpha^{+}_{2i-j}(0), & i=1,2,\dots,N, \ 1\leq j \leq \min \{2i,N\},\\
 0, & \min \{2i,N\} < j \leq N.
\end{cases}
\end{gather*}
Now it is evident that each element in \eqref{fN'} is homogenous with degree
\begin{gather*}\mathcal{D}[f^{'}_{ij}]= \begin{cases}
 2i-j, & i=1,2,\dots,N, \ 1\leq j \leq \min \{2i,N\},\\
 {\rm not~available}, & \min \{2i,N\} < j \leq N.
\end{cases}
\end{gather*}
Note that the degree of nonzero $f^{'}_{ij}$ is separable in terms of $i$ and $j$. Meanwhile, $f_N$ is an algebraic summation in which each term is a nonzero product $\prod^{N}_{i=1} f^{'}_{i,j_i}$ where $j_i$ runs over a~permutation of the set $\{1, 2, \dots, N\}$. It is easy to get{\samepage
\begin{gather*}\mathcal{D}\left[\prod^{N}_{i=1} f^{'}_{i,j_i}\right]=\sum^N_{i=1}(2i)-\sum^N_{j=1} j=\frac{N(N+1)}{2},\end{gather*}
which means $f_N$ is homogeneous with degree $\mathcal{D}[f_N]=\frac{N(N+1)}{2}$.}

\looseness=-1 (ii) In the following we come to the statement that $f_N$ depends only on $\{x_1, x_3, \dots, x_{2N-1}\}$, which has been shown correct for $N=1,2,3$ in~\eqref{f-N}. Now we assume the statement is correct for $f_j$ where $j\in \{1,2,\dots, N-1\}$. For $f_N$, f\/irst, it contains $x_{2N-1}$. In fact, from the expression~\eqref{fN'} we can see that $x_{2N-1}$ only appears in the element $f^{'}_{N,1}$, and doing Laplace expansion for \eqref{fN'} along the f\/irst two columns it is easy to f\/ind the only term involving $x_{2N-1}$ is
\begin{gather}
\frac{(-1)^{N+1}}{2N-1} x_{2N-1} f_{N-2}.\label{x-2N-1}
\end{gather}
Next, we note that for each monomial $A=\prod_{i\geq 1} x_i^{j_i}$, the commutating relation
\begin{gather*}\partial_{x_i}\t A = \t {\partial_{x_i}A}
\end{gather*}
holds, which indicates $\partial_{x_i}\t f_N = \t {\partial_{x_i}f_N}$. Taking derivative $\partial_{x_{2i}}$ on both sides of double down bar shifted \eqref{f-itera-a} we have
\begin{gather*} f_{N-2} \t{\partial_{x_{2i}}f_N}-\t f_{N-2} \partial_{x_{2i}}f_{N}=0,\end{gather*}
where the right hand side has vanished due to the assumption that $f_j$ is independent of $x_{2i}$ for $j\in\{1,2,\dots, N-1\}$. The above relation indicates $\frac{\partial_{x_{2i}}f_{N}}{f_{N-2}}$ is independent of~$n$. In a same manner, from \eqref{f-itera-b} we can f\/ind $\frac{\partial_{x_{2i}}f_{N}}{f_{N-2}}$ is independent of~$m$. Thus we come to a relation
\begin{gather}
\partial_{x_{2i}}f_{N}=c f_{N-2},\label{f-x2i}
\end{gather}
where $c$ is a constant. Based on item (i) of the current theorem the degrees in \eqref{f-x2i} read
\begin{gather*}\frac{N(N+1)}{2}-2i= \frac{(N-2)(N-1)}{2},\end{gather*}
i.e.,
\begin{gather*}2N-1=2i,\end{gather*}
which is contradictory for any $i\in \mathbb{Z}^+$ and then indicates $c=0$ in \eqref{f-x2i}. Thus, $\partial_{x_{2i}}f_{N}=0$, i.e., $f_{N}$ is independent of $x_{2i}$. In conclusion, $f_N$ depends only on $\{x_1, x_3, \dots, x_{2N-1}\}$.

(iii) 
It is known that $v_N=f_{N}/f_{N-2}$ and $V_N=f_{N-1}/f_{N-2}$ provide a solution pair to~\eqref{BT-S-mkdv-N}. Meanwhile, from Lemma~\ref{lem-3} both $v_N$ and $V_N$ are positive in the f\/irst quadrant $\{n\geq 0,\, m\geq 0\}$ if $a>0$, $b>0$ and $v_N(0,0)>0$. Obviously, from $v_N=f_{N}/f_{N-2}$ and $f_{-1}=f_0=1$, for each $N=2,3,\dots$, we can successively f\/ind $f_N(n,m)>0$ in the f\/irst quadrant if $a>0$, $b>0$ and $f_N(0,0)>0$.

(iv) Finally, we formulate a mechanism to guarantee $f_N(0,0)>0$ by choosing suitable~$\gamma_{j}$. From item~(ii) we know that $f_N(0,0)$ is only related to $\{\gamma_1,\gamma_3,\dots,\gamma_{2N-1}\}$. Since in $f_N$ the term involving $x_{2N-1}$ is \eqref{x-2N-1}, we express $f_N$ as
\begin{gather*} f_{N}(n,m)=f_{N}(n,m)|_{x_{2N-1}=0}+\frac{(-1)^{N+1}}{2N-1}x_{2N-1}f_{N}(n,m), \end{gather*}
which, at point $(0,0)$, yields
\begin{gather*} f_N(0,0)=f_{N}(0,0)|_{\gamma_{2N-1}=0}+\frac{(-1)^{N+1}}{2N-1}\gamma_{2N-1}f_{N-2}(0,0).\end{gather*}
To guarantee $f_N(0,0)>0$, we need to take
\begin{gather}
(-1)^{N+1}\gamma_{2N-1}> -\frac{(2N-1)f_{N}(0,0)|_{\gamma_{2N-1}=0}}{f_{N-2}(0,0)}.\label{ga}
\end{gather}
Thus, as a starting step we take $f_0=1$ and $f_1=x_1$ with $\gamma_1>0$, for $N=2$, using the above formula we choose a value for $\gamma_3$ so that $f_2(0,0)>0$. We can repeatedly use~\eqref{ga} and will successively choose $\gamma_{2N-1}$ and get $f_N(0,0)>0$ for higher~$N$.
\end{proof}

\subsection*{Acknowledgements}
We are grateful to the referee for the invaluable comments. This project is supported by the NSF of China (no.~11371241 and no.~11631007).

\pdfbookmark[1]{References}{ref}
\LastPageEnding

\end{document}